\documentclass[oneside]{amsart}
\usepackage{lmodern}
\usepackage{lmodern}
\usepackage[T1]{fontenc}
\usepackage{textcomp}
\usepackage[utf8]{inputenc}
\usepackage{amsbsy}
\usepackage{amstext}
\usepackage{amsthm}
\usepackage{amssymb}
\usepackage{graphicx}
\usepackage{geometry}
\usepackage[bookmarks=false,
 breaklinks=false,pdfborder={0 0 1},backref=section,colorlinks=false]
 {hyperref}
\hypersetup{
 hidelinks,pdfcreator={LaTeX via pandoc}}

\makeatletter

\newcommand*\LyXZeroWidthSpace{\hspace{0pt}}

\numberwithin{equation}{section}
\numberwithin{figure}{section}
\theoremstyle{remark}
\newtheorem*{rem*}{\protect\remarkname}
\theoremstyle{plain}
\newtheorem*{thm*}{\protect\theoremname}
\theoremstyle{plain}
\newtheorem{thm}{\protect\theoremname}
\theoremstyle{remark}
\newtheorem{rem}[thm]{\protect\remarkname}
\theoremstyle{plain}
\newtheorem*{lem*}{\protect\lemmaname}

\@ifundefined{date}{}{\date{}}

\usepackage{xcolor}
\usepackage{iftex}
\ifPDFTeX
  \usepackage{textcomp}
\else 
  \usepackage{unicode-math}
  \defaultfontfeatures{Scale=MatchLowercase}
  \defaultfontfeatures[\rmfamily]{Ligatures=TeX,Scale=1}
\fi
\ifPDFTeX\else
\fi
\IfFileExists{upquote.sty}{\usepackage{upquote}}{}
\IfFileExists{microtype.sty}{
  \usepackage[]{microtype}
  \UseMicrotypeSet[protrusion]{basicmath} 
}{}

\@ifundefined{KOMAClassName}{
  \IfFileExists{parskip.sty}{%
    \usepackage{parskip}
  }{
    \setlength{\parindent}{0pt}
    \setlength{\parskip}{6pt plus 2pt minus 1pt}}
}{
  \KOMAoptions{parskip=half}}

\usepackage{bookmark}
\IfFileExists{xurl.sty}{\usepackage{xurl}}{} 
\author{}

\makeatother

\providecommand{\lemmaname}{Lemma}
\providecommand{\remarkname}{Remark}
\providecommand{\theoremname}{Theorem}

\begin{document}
\title{Popularity Bias Alignment Estimates}
\author{Anton Lyubinin}
\date{November 24, 2025}
\begin{abstract}
We are extending Popularity Bias Memorization theorem from \cite{LGCZHFCW}
in several directions. We extend it to arbitrary degree distributions
and also prove both upper and lower estimates for the alignment with
``top-k'' singular hyperspace.
\end{abstract}

\maketitle

\section*{Introduction}

Recommender systems routinely exhibit \textbf{popularity bias}: items
that already attract many interactions receive disproportionate exposure,
which further entrenches their dominance. Recent spectral analyses
formalize one mechanism behind this effect by studying how a model's
recommendation scores align with singular vectors of the user-item
interaction matrix. In particular, Lin et al. in \cite{LGCZHFCW}
have studied popularity bias in collaborative filtering embedding-based
recommender systems, under the assumption that the item-popularity
distribution follows a power law. Following their empirical observations,
they have established Popularity Bias Memorization Effect, showing
alignment of popularity with principal singular vector and obtaining
explicit bounds on the alignment between model outputs and this direction,
and Popularity Bias Amplification Effect, finding the bound for the
relative proportion of most popular item in top-1 recommendation over
all users. They have also proposed a debiasing method based on their
findings.

However, power laws are not the only possible item popularity distribution
in real data. A power law appears as a straight line on a log-log
plot, but many other monotone decreasing distributions can masquerade
as ``straight enough'' under common plotting choices. Empirically,
popularity counts across items often deviate from pure power laws:
double-Pareto-lognormal (DPLN) families have been proposed on theoretical
and empirical grounds \cite{RJ,FWLG}, and multiple datasets in social/news
domains display shapes better captured by log-normal, power-law-with-cutoff,
or other heavy-tailed forms \cite{HL,LG,CPM}. 

\noindent\begin{minipage}[t]{1\columnwidth}%
\includegraphics[scale=0.43]{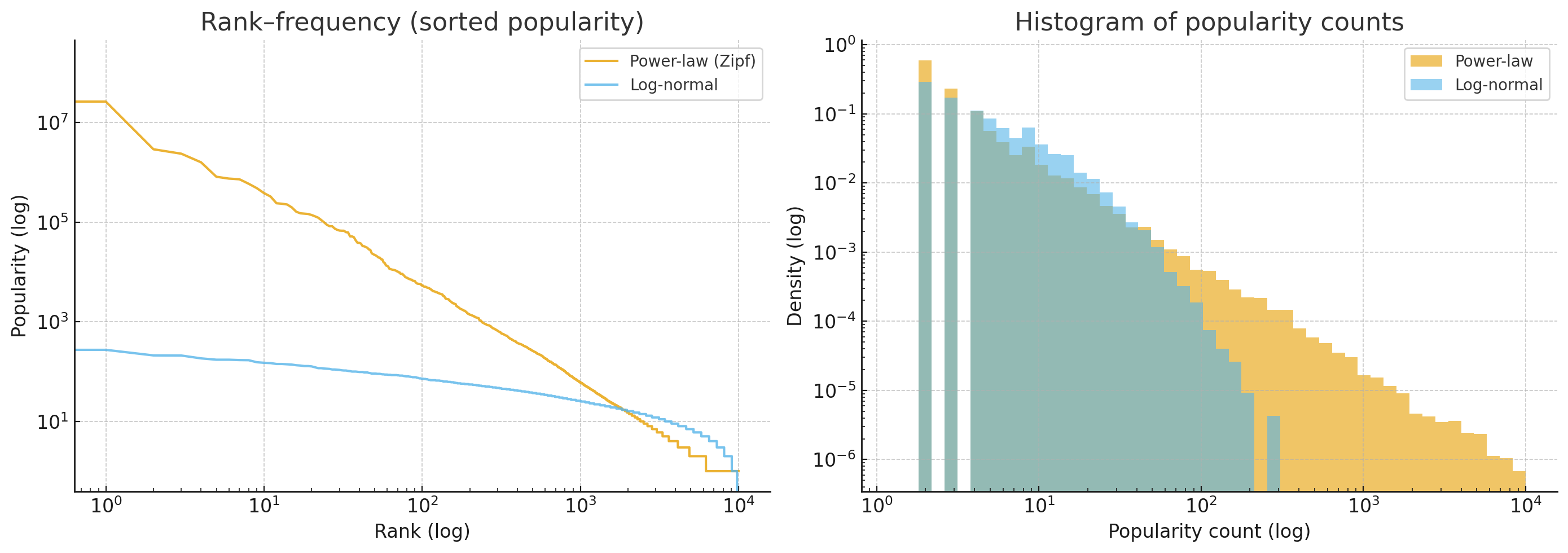}

Log-normal: $\mu=2.0$, $\sigma=1.0$, Power law: $\alpha=1.5$%
\end{minipage}

Importantly, a log-normal distribution of item counts does not induce
a power law after ranking. Moreover, goodness-of-fit is tricky: depending
on parameters, both log-normal and even exponential tails can look
deceptively straight on logarithmic axes, and statistical tests, like
the Kolmogorov--Smirnov test, have limited power to reliably distinguish
these alternatives from power laws \cite{CSN}. Consistent with this,
Clauset--Shalizi--Newman emphasize that pure power laws are rare
in practice, while log-normal and power-law-with-cutoff often fit
better \cite{CSN}.

These observations matter for spectral analyses of popularity bias.
If the data's popularity distribution deviates from a pure power law,
alignment with the principal singular vector need not be strong -
and may vary systematically with the tail shape and truncation. This
motivates a study that does not assume a specific parametric tail.

Building on the spectral perspective from \cite{LGCZHFCW}, we remove
the power-law assumption and establish distribution-agnostic results:
\begin{enumerate}
\item \textbf{$\Pi_{1}$-memorization.} We extend the Popularity Bias Memorization
theorem to an arbitrary item--popularity distribution, identifying
conditions under which recommendation scores align with the principal
singular vector. 
\item \textbf{$\Pi_{k}$-memorization.} We prove several two-sided bounds
for the alignment between popularity and the top-\textbf{$k$} singular
``hyperspace'' (i.e., the subspace spanned by the top $k$ right
singular vectors). These bounds recover the power-law case of \cite{LGCZHFCW}
as a special instance and reveal regimes where alignment weakens (e.g.,
log-normal or truncated/exponential-like tails). The very existence
of the upper bound on the alignment means that in some cases the alignment
with top-\textbf{$k$} singular hyperspace (and, in particular, with
the principal singular vector) has limits. 
\end{enumerate}
Our framework yields interpretable quantities---expressed in terms
of tail mass, effective ranks, and concentration---that can be estimated
from data and used to anticipate when models will (or will not) align
with popularity. Let us outline the content of this paper.

In \textbf{Preliminaries} section we review related notions from graph
theory and collaborative filtering models.\textbf{ Popularity bias
memorization }contains the statement of $\Pi_{1}$-memorization theorem
with complete proof, and the statements and discussion of results
on $\Pi_{k}$-memorization. The proofs of $\Pi_{k}$-memorization,
due to their technical nature and complexity, are moved to the \textbf{Appendix},
together with our discussion of the interpretation of Kumar's bound
in terms of graph topology.

\section{Preliminaries}

We will be considering collaborative filtering embedding-based recommender
systems. This section fixes notation and reviews related notions about
bipartite graphs (user--item), its linear-algebraic representations,
and basic spectral facts used throughout.

\subsubsection*{Bipartite interaction graph and matrices}
\begin{itemize}
\item Let $G=(V_{G},E_{G})$ be a bipartite graph with vertex partition
$V_{G}=\{\mathcal{U},\mathcal{I}\}$, where $\mathcal{U}$ are users
and $\mathcal{I}$ are items, $|\mathcal{U}|=n$, $|\mathcal{I}|=m$. 
\item The binary interaction matrix (biadjacency matrix) is $Y\in\{0,1\}^{n\times m}$,
with $(u,i)\in E_{G}\iff Y_{ui}=1$, represents interactions between
users and items. 
\item the block matrix $A$ is the (bipartite) adjacency matrix of $G$,
\[
A=\begin{bmatrix}0 & Y\\
Y' & 0
\end{bmatrix}\in\mathbb{R}^{(n+m)\times(n+m)},
\]
\item We use prime $(\cdot)'$ for transpose; all vectors are columns. 
\item Degrees and ``popularity.'' For an item $i\in\mathcal{I}$: 
\[
r_{i}=\sum_{u\in\mathcal{U}}Y_{ui}\quad\text{(item degree / popularity).}
\]
\item Let $\vec{e}_{k}=(1,\dots,1)'\in\mathbb{R}^{k}$ and define the popularity
vector: 
\[
\vec{r}=(r_{i})_{i=1}^{m}=Y'\vec{e}_{n}\in\mathbb{R}^{m}.
\]
\item For a user $u\in\mathcal{U}$, $d_{u}=\sum_{i}Y_{ui}$ and for an
item $i$, $d_{i}=r_{i}$. - Edge count and volumes: 
\[
|E_{G}|=\sum_{u,i}Y_{ui}=\vec{e}_{n}'Y\vec{e}_{m}=\vec{e}_{m}'\vec{r}.
\]
\end{itemize}

\subsubsection*{Volumes and degree summaries (following \cite{CL})}
\begin{itemize}
\item For any $S\subseteq V_{G}$ and $k\ge0$, with $d_{x}$ the degree
of vertex $x$: 
\[
vol_{k}(S)=\sum_{x\in S}d_{x}^{k},\quad vol(S)=vol_{1}(S),\quad|S|=vol_{0}(S),\quad\bar{d}_{S}=\frac{vol(S)}{|S|},\quad\tilde{d}_{S}=\frac{vol_{2}(S)}{vol(S)}.
\]
\item Global summaries:
\[
|G|=|V_{G}|=vol_{0}(G),\qquad\bar{d}=\frac{vol(G)}{|G|},\qquad\tilde{d}=\frac{vol_{2}(G)}{vol(G)}.
\]
\item For the item side in a recommender graph: 
\[
r_{i}=d_{i},\qquad\|\vec{r}\|_{2}^{2}=\sum_{i=1}^{m}r_{i}^{2}=vol_{2}(\mathcal{I}).
\]
\end{itemize}

\subsubsection*{Embedding models and low-rank structure (following \cite{LGCZHFCW,WHWZW})}
\begin{itemize}
\item Embedding-based recommenders assign $d$-dimensional representations
to users/items: 
\[
U\in\mathbb{R}^{n\times d},\quad V\in\mathbb{R}^{m\times d},\quad\text{rows: }u_{u}'\text{ and }v_{i}'.
\]
\item A generic link function $\mu:\mathbb{R}\to\mathbb{R}$ maps inner
products to scores: 
\[
\hat{y}_{ui}=\mu(u_{u}'v_{i}),\qquad\hat{Y}=(\hat{y}_{ui})=\mu(UV').
\]
Note that $\mu$ is applied entrywise, with mean-squared error: 
\[
L_{R}=\|Y-\hat{Y}\|_{F}^{2}\;\to\;\min.
\]
\item When $\mu$ is the identity (or after an appropriate inverse link/linearization),
the model reduces to finding a low-rank matrix $\hat{Y}=UV'$.
\end{itemize}

\subsubsection*{SVD and Eckart--Young--Mirsky (EYM)}

Let the singular value decomposition (SVD) of $Y$ be: 
\[
Y=P\Sigma Q'=\sum_{k=1}^{\text{rank}(Y)}\sigma_{k}p_{k}q_{k}',
\]
with $\{p_{k}\}_{k=1}^{n}$ and $\{q_{k}\}_{k=1}^{m}$ orthonormal
singular vectors and singular values $\sigma_{1}\ge\cdots\ge0$.

By the Eckart--Young--Mirsky theorem, the best rank-$\ell$ approximation
(in any unitarily invariant norm, in particular the Frobenius norm)
is the truncated SVD: 
\[
\hat{Y}_{\ell}=\sum_{k=1}^{\ell}\sigma_{k}p_{k}q_{k}'.
\]
Similar to \cite{LGCZHFCW}, we will be assuming that after the training
$||Y-\hat{Y}||_{F}^{2}$ (which, in general, depends on $\mu$) achieves
the minimum, and so $\hat{Y}=\hat{Y}_{\ell}$ for some $\ell$. Equivalent
eigen-relations: 
\[
Yq_{k}=\sigma_{k}p_{k},\qquad Y'p_{k}=\sigma_{k}q_{k},
\]
\[
(YY')p_{k}=\sigma_{k}^{2}p_{k},\qquad(Y'Y)q_{k}=\sigma_{k}^{2}q_{k}.
\]
Hence, for the truncated approximant $\hat{Y}_{\ell}$, the nonzero
singular values are $\{\sigma_{1},\dots,\sigma_{\ell}\}$; equivalently:
\[
\sigma_{i}(\hat{Y}_{\ell})=\sigma_{i}(Y)\;\;\text{for }i=1,\dots,\ell,\qquad\sigma_{i}^{2}(Y)=\lambda_{i}(YY')=\lambda_{i}(Y'Y).
\]

\begin{rem*}
If $Y'Yv=a^{2}v$ with $a>0$ and $u=\tfrac{1}{a}Yv$, then $Yv=au$
and $Y'u=av$; thus eigenvectors of $Y'Y$ (resp. $YY'$) are right
(resp. left) singular vectors of $Y$.
\end{rem*}

\subsubsection*{Spectrum of the bipartite adjacency}

For the bipartite adjacency matrix: 
\[
A=\begin{bmatrix}0 & Y\\
Y' & 0
\end{bmatrix},
\]
the spectrum is symmetric: 
\[
\text{spec}(A)=\{\pm\sigma_{k}:k=1,\dots,\text{rank}(Y)\}\;\cup\;\{0\;\text{with appropriate multiplicity}\}.
\]
In particular, the nonzero eigenvalues of $A$ are $\pm$ the singular
values of $Y$.

\section{Popularity bias memorization.}

\subsection{$\Pi_{1}$-memorization.}

Below is a direct generalization of the Popularity Bias Memorization
theorem from \cite{LGCZHFCW} to the case of an arbitrary recommender
graph.
\begin{thm*}
\textbf{$\boldsymbol{1A.}$} In the above notations, 
\[
A:\quad\cos(\vec{r},\vec{q_{1}})\geq\frac{\sigma_{1}^{2}}{vol_{2}(\mathcal{I})}\sqrt{1-\frac{1}{\sigma_{1}^{2}}(vol(\mathcal{I})-r_{max})}.
\]
\[
B:\quad\sigma_{1}^{2}\geq r_{max}.
\]
\end{thm*}
\begin{proof}
A: The proof in general follows \cite{LGCZHFCW}. 
\[
\cos(\vec{r},\vec{q_{1}})=\frac{(Y'\vec{e})'\vec{q_{1}}}{||\vec{r}||}=\frac{\vec{e}'(Y\vec{q_{1}})}{||\vec{r}||}=\frac{\sigma_{1}\vec{e}'\vec{p}_{1}}{\sqrt{vol_{2}(\mathcal{I})}}
\]
\[
\cos(\vec{r},\vec{q_{s}})=\frac{(Y'\vec{e})'\vec{q_{s}}}{||\vec{r}||}=\frac{\vec{e}'(Y\vec{q_{s}})}{||\vec{r}||}=\frac{\sigma_{s}\vec{e}'\vec{p}_{s}}{\sqrt{vol_{2}(\mathcal{I})}}
\]
If $Z=Y'Y=||z_{kl}||$, $z_{kl}\geq0$, then 
\[
\sigma_{1}^{2}=\max_{||\vec{q}||=1}\vec{q}'(Y'Y)\vec{q}=\max_{||\vec{q}||=1}\sum_{k=1}^{m}\sum_{l=1}^{m}q_{k}z_{kl}q_{l}
\]
As explained in \cite{LGCZHFCW}, it is possible to chose $\vec{q}_{1}$
so that $\vec{q}_{1}\geq0$ ($q_{1i}>0$), then $\vec{p}_{1}=\frac{1}{\sigma_{1}}(Y'\vec{q}_{1})\geq0$.
Note that all other eigenvectors (and so singular vectors) may have
at least one negative component (see Perron-Frobenius theorem).

Let $s$ be the number of an item, $\vec{y}_{k}=||y_{\cdot,k}||$,
$\vec{p}_{1}\geq0$, $y_{u,i}\in\{0,1\}$. Since $Y'\vec{p}_{1}=\sigma_{1}\vec{q}_{1}$,
we have $\vec{y}_{s}'\vec{p}_{1}=\sigma_{1}q_{s1}$ and $\vec{e}'\vec{p}_{1}\geq\vec{y}_{s}'\vec{p}_{1}=\sigma_{1}q_{s1}$.
So we get the initial lower bound, 
\[
\cos(\vec{r},\vec{q_{1}})\geq\frac{\sigma_{1}^{2}q_{s1}}{\sqrt{vol_{2}(\mathcal{I})}}.
\]
We shall try to estimate the $q_{s1}$ term through $q_{s1}^{2}=1-\sum_{i\neq s}q_{i1}^{2}$,
since $Q$ is unitary matrix.

If $B_{s}=Y$ without $\vec{y}_{s}$, then for any $j$ 
\[
\sum_{i\neq s}\vec{y}_{j}\vec{y}_{i}\leq\sum_{i\neq s}\vec{y}_{i}\vec{y}_{i}=\sum_{i\neq s}r_{i}=vol(\mathcal{I})-r_{s}
\]
Then for $C_{s}=B_{s}'B_{s}=||c_{ij}||$ we have $\sum_{j}c_{ij}\leq vol(\mathcal{I})-r_{s}$,
and so, by Perron-Frobenius theorem, 
\[
\lambda_{max}(B_{s}'B_{s})\leq\max\sum_{j}c_{ij}\leq vol(\mathcal{I})-r_{s}.
\]
This gives us an estimate of ${\displaystyle \sum_{i\neq s}(q_{i1})^{2}}$,
\[
\sum_{i\neq s}(q_{i1})^{2}=\frac{1}{\sigma_{1}^{2}}||B_{s}'\vec{p}_{1}||^{2}\leq\frac{1}{\sigma_{1}^{2}}\lambda_{max}(B_{s}'B_{s})\leq\frac{1}{\sigma_{1}^{2}}(vol(\mathcal{I})-r_{s}).
\]
Thus we get 
\[
q_{s1}^{2}=1-\sum_{i\neq s}q_{1i}^{2}\geq1-\frac{1}{\sigma_{1}^{2}}(vol(\mathcal{I})-r_{s})
\]
and so 
\[
\cos(\vec{r},\vec{q_{1}})\geq\frac{\sigma_{1}^{2}}{\sqrt{vol_{2}(\mathcal{I})}}\sqrt{1-\frac{1}{\sigma_{1}^{2}}(vol(\mathcal{I})-r_{s})},
\]
and, in particular, 
\[
\cos(\vec{r},\vec{q_{1}})\geq\frac{\sigma_{1}^{2}}{\sqrt{vol_{2}(\mathcal{I})}}\sqrt{1-\frac{1}{\sigma_{1}^{2}}(vol(\mathcal{I})-r_{max})}.
\]
B: same as in \cite{LGCZHFCW}. 
\end{proof}
\begin{rem}
i. If item degrees or rank-frequencies follow distribution $\rho$,
define $vol(\rho)$ and $vol_{2}(\rho)$ in similar way (assuming
they exist, i.e.~the required sums converge). Then $vol(\rho)>vol(\mathcal{I})$
and $vol_{2}(\rho)>vol_{2}(\mathcal{I})$ and so we get a new estimate:
\[
\cos(\vec{r},\vec{q_{1}})\geq\frac{\sigma_{1}^{2}}{\sqrt{vol_{2}(\mathcal{I})}}\sqrt{1-\frac{1}{\sigma_{1}^{2}}(vol(\mathcal{I})-r_{max})}\geq\frac{\sigma_{1}^{2}}{\sqrt{vol_{2}(\rho)}}\sqrt{1-\frac{1}{\sigma_{1}^{2}}(vol(\rho)-r_{max})}.
\]
This the type of estimate from \cite{LGCZHFCW}.

ii. In order to analyze a popularity distribution on alignment with
principal singular vector one need to estimate further the part of
the RHS in theorem A. One can use the estimates for radical and the
singular number to simplify it. The standard inequalities for estimating
the square root are $\sqrt{1-a}\geq1-\frac{a}{2}-\frac{a^{2}}{2}\geq1-a$.
If $\sigma_{1}^{2}\geq L$, then we get the following estimates: 
\[
\begin{aligned}\cos(\vec{r},\vec{q}_{1}) & \ge\frac{L}{\sqrt{\,vol_{2}(\mathcal{I})\,}}\sqrt{1-\frac{1}{L\!}\bigl(vol(\mathcal{I})-r_{\max}\bigr)}\ge\\
 & \ge\frac{L}{\sqrt{\,vol_{2}(\mathcal{I})\,}}\left(1-\frac{vol(\mathcal{I})-r_{\max}}{2L}-\frac{\bigl(vol(\mathcal{I})-r_{\max}\bigr)^{2}}{2L^{2}}\right)\ge\\
 & \ge\frac{L}{\sqrt{\,vol_{2}(\mathcal{I})\,}}\left(1-\frac{vol(\mathcal{I})-r_{\max}}{L}\right).
\end{aligned}
\]
The rough estimates for $\sigma_{1}^{2}$ are $u_{min}\leq\sigma_{1}\leq u_{max}$.
There is more refined estimate by Kumar \cite{K}, that we review
in the appendix, 
\[
\sigma_{1}^{2}\geq K_{L}\simeq\bar{d}+\frac{s_{K}}{\sqrt{p-1}},\quad s_{K}^{2}\simeq Var(d)+\frac{W}{p}+\frac{4B}{p},
\]
where 
\end{rem}

\begin{itemize}
\item $s_{K}$ is the ``upscaled variance'', 
\item $p=\left\lfloor \frac{n+m}{2}\right\rfloor $, 
\item $W=W(G)$ is the ``wedge number'', 
\item $B=c_{4}(G)$ is the ``butterfly number'' of the graph $G$.
\end{itemize}
Kumar's lower bound $K_{L}$ depends on the topology of the whole
graph $G$, showing deep connection between this topology and popularity
alignment.

\subsection{$\Pi_{k}$-memorization.}

For a general popularity distribution we may not necessarily have
an alignment with the principal singular vector. In this section we
will prove that we always have a closer alignment with the span of
the top-$k$ singular vectors. We will prove several bounds, each
with their own pros/cons.

Let $\Pi_{k}$\hspace{0pt} be the orthogonal projector onto the span
of the first $k$ right singular vectors $\{\vec{q_{1}},\dots,\vec{q_{k}}\}$
and $Q_{k}=[\vec{q_{1}},\dots,\vec{q_{k}}]$ is the matrix formed
by them. Define the angle between the span and the popularity vector,

\[
\theta_{k}:=\angle\big(r,\ \Pi_{k}\vec{r}\big).
\]

Then 
\[
\cos\theta_{k}\;=\;\frac{\|\Pi_{k}\vec{r}\|}{\|\vec{r}\|}\;=\sqrt{\frac{\vec{r}'\Pi_{k}\vec{r}}{||\vec{r}||^{2}}}\;=\;\sqrt{\kappa_{k}},\qquad\Pi_{k}\;=\;Q_{k}Q_{k}',\qquad\kappa_{k}:=\sum_{i=1}^{k}\cos^{2}(r,\vec{q}_{i}).
\]

\subsubsection*{Combinatorial estimate.}

Our proof will use the strategy from $k=1$, so in preparation we
will prove a few lemmas. Let $S\subset\mathcal{I}$ be any set of
items, $B_{S}$\hspace{0pt} be $Y$ with the columns in $S$ removed,
and $\vec{v}_{S}$ be a projection of a vector $\vec{v}$ onto the
span of $\{\vec{q}_{s}\}_{s\in S}$. Also let 
\[
\tau_{S}=\sum_{j,s\in S}q_{js}^{2}=\sum_{j=1}^{k}\|(\vec{q}_{j})_{S}\|_{2}^{2}\;,\qquad\text{and}\qquad r_{S}=\sum_{r\in S}r_{s}.
\]

\begin{lem*}
1. With $\tau_{S}$ as above,
\[
\tau_{S}\geq k-(vol(\mathcal{I})-r_{S})\left(\sum_{j=1}^{k}\frac{1}{\sigma_{j}^{2}}\right)=k-\Delta_{S}H_{k},
\]
where $\Delta_{S}\LyXZeroWidthSpace:=vol(I)-r_{S}$\hspace{0pt},
and $H_{k}:=\sum_{j=1}^{k}\sigma_{j}^{-2}$.
\end{lem*}
The bound of the lemma can be improved using Ky Fan's principle. 
\begin{lem*}
$1'$. 
\[
k-\frac{1}{\sigma_{1}^{2}}\sum_{j=n-k+1}^{n}\lambda_{j}(B_{S}B_{S}')\;\geq\;\tau_{S}\;\ge\;k-\frac{1}{\sigma_{k}^{2}}\sum_{j=1}^{k}\lambda_{j}(B_{S}'B_{S}).
\]
\end{lem*}
\begin{thm*}
$\boldsymbol{1C.}$ Let $U_{\tau_{S}}\geq\tau_{S}\geq L_{\tau_{S}}$.
Then 
\[
\vec{r}'\Pi_{k}\vec{r}\;\geq\;\Big(\,\big[L_{\tau_{S}}-(|S|-1)\big]_{+}\Big)\,\|\vec{r}_{S}\|^{2}\;-\;2\,\|\vec{r}_{S}\|\,\|\vec{r}_{S^{c}}\|,
\]
 
\[
\vec{r}'\Pi_{k}\vec{r}\;\leq\ \min\!\Big\{1,\ U_{\tau_{S}}\Big\}\,\|r_{S}\|^{2}+2\|r_{S}\|\,\|r_{S^{c}}\|+\|r_{S^{c}}\|^{2}.
\]
\end{thm*}
Normalizing the inequality from the theorem by $\|\vec{r}\|=\sqrt{\operatorname{vol}_{2}(\mathcal{I})}$\hspace{0pt},
for the values of $L_{\tau_{S}}$, $U_{\tau_{S}}$ from lemmas 1,
1', we are getting the following bounds for the angle:

\textbf{Combinatorial bounds:}

\textbf{A1. }For general $|S|$, using $L_{\tau_{S}}=k-\Delta_{S}H_{k}$,
with $\Delta_{S}\LyXZeroWidthSpace:=vol(I)-r_{S}$\hspace{0pt}, and
$H_{k}:=\sum_{j=1}^{k}\sigma_{j}^{-2}$, 
\[
\quad\cos\theta_{k}\;=\;\frac{\|\Pi_{k}\vec{r}\|}{\|\vec{r}\|}\;\ge\;\frac{\Big(\,[\,k-\Delta_{S}H_{k}-(|S|-1)\,]_{+}\ \|\vec{r}_{S}\|^{2}\ -\ 2\,\|\vec{r}_{S}\|\,\|\vec{r}_{S^{c}}\|\Big)_{+}^{1/2}}{\ \sqrt{\operatorname{vol}_{2}(\mathcal{I})}\ }\;
\]

\textbf{A2.} For $|S|=k$, A1 becomes 
\[
\quad\cos\theta_{k}\;=\;\frac{\|\Pi_{k}\vec{r}\|}{\|\vec{r}\|}\;\ge\;\frac{\Big(\,[\,1-\Delta_{S}H_{k}\,]_{+}\ \|\vec{r}_{S}\|^{2}\ -\ 2\,\|\vec{r}_{S}\|\,\|\vec{r}_{S^{c}}\|\Big)_{+}^{1/2}}{\ \sqrt{\operatorname{vol}_{2}(\mathcal{I})}\ }\;
\]

\begin{rem}
If $S$ = ``\textbf{$k$ most popular items}'' (largest $r_{s}$\hspace{0pt}),
then $r_{S}$\hspace{0pt} and $\|r_{S}\|^{2}=\sum_{s\in S}r_{s}^{2}$
\hspace{0pt} are large, and the deficit $\Delta_{S}=\operatorname{vol}(\mathcal{I})-r_{S}$\hspace{0pt}
is small - tightening the bound. As $k$ grows, $\Delta_{S}$\hspace{0pt}
shrinks and $H_{k}$\hspace{0pt} grows slowly; the factor $[1-\Delta_{S}H_{k}]_{+}$\hspace{0pt}
increases, so the bound strengthens---matching the intuition that
a larger top subspace captures more of $r$.
\end{rem}

\textbf{B1.} \textbf{Ky Fan strengthening.} Using $L_{\tau_{S}}=k-\frac{1}{\sigma_{k}^{2}}\sum_{j=1}^{k}\lambda_{j}(B_{S}'B_{S})$
\[
\cos\theta_{k}\;=\;\frac{\|\Pi_{k}\vec{r}\|}{\|\vec{r}\|}\;\geq\;\sqrt{\frac{\Big(\,\Big[k-\frac{1}{\sigma_{k}^{2}}\sum_{j=1}^{k}\lambda_{j}(B_{S}'B_{S})-(|S|-1)\Big]_{+}\Big)\,\|\vec{r}_{S}\|^{2}\;-\;2\,\|\vec{r}_{S}\|\,\|\vec{r}_{S^{c}}\|}{\mathrm{vol}_{2}(\mathcal{I})}}
\]

\textbf{B2.} For $|S|=k$, B1 becomes 
\[
\cos\theta_{k}\;=\;\frac{\|\Pi_{k}\vec{r}\|}{\|\vec{r}\|}\;\geq\;\sqrt{\frac{\Big(\,\Big[1-\frac{1}{\sigma_{k}^{2}}\sum_{j=1}^{k}\lambda_{j}(B_{S}'B_{S})\Big]_{+}\Big)\,\|\vec{r}_{S}\|^{2}\;-\;2\,\|\vec{r}_{S}\|\,\|\vec{r}_{S^{c}}\|}{\mathrm{vol}_{2}(\mathcal{I})}}
\]

These are the bounds A1/A2 with the crude $\lambda_{\max}$ in the
middle steps (see appendix) replaced\hspace{0pt} by the \emph{Ky
Fan sum} $\sum_{j=1}^{k}\lambda_{j}(B_{S}'B_{S})$. As the result,
it strictly strengthens A1/A2 (since $\sum_{j\leq k}\lambda_{j}\leq k\lambda_{{\rm max}}$).
Still combinatorial (uses which columns are removed).

\textbf{B3.} \textbf{Cauchy interlacing simplification.} Using Cauchy
interlacing, 
\[
\sigma_{i}^{2}=\lambda_{i}(YY')\;\ge\;\lambda_{i}(B_{S}B_{S}')\;\ge\;\lambda_{i+|S|}(YY')=\sigma_{i+|S|}^{2},
\]
we can get the bound in terms of the singular values of $Y$. We get
from B1 
\[
\cos\theta_{k}\;=\;\frac{\|\Pi_{k}\vec{r}\|}{\|\vec{r}\|}\;\geq\;\sqrt{\frac{\Big[k-\frac{1}{\sigma_{k}^{2}}\left(\sum_{j=1}^{k}\sigma_{j}^{2}\right)-(|S|-1)\Big]_{+}\,\|\vec{r}_{S}\|^{2}\;-\;2\,\|\vec{r}_{S}\|\,\|\vec{r}_{S^{c}}\|}{\mathrm{vol}_{2}(\mathcal{I})}}
\]
and from B2 
\[
\cos\theta_{k}\;=\;\frac{\|\Pi_{k}\vec{r}\|}{\|\vec{r}\|}\;\geq\;\sqrt{\frac{\,\Big[1-\frac{1}{\sigma_{k}^{2}}\left(\sum_{j=1}^{k}\sigma_{j}^{2}\right)\Big]_{+}\,\|\vec{r}_{S}\|^{2}\;-\;2\,\|\vec{r}_{S}\|\,\|\vec{r}_{S^{c}}\|}{\mathrm{vol}_{2}(\mathcal{I})}}
\]
The bounds B3 become basis-invariant and easier to evaluate, but \textbf{weaker}
than B1/B2.

\textbf{Upper bounds. C1.} We have the general upper bound 
\[
\cos\theta_{k}\;=\;\frac{\|\Pi_{k}\vec{r}\|}{\|\vec{r}\|}\;\leq\;\sqrt{\frac{\lambda_{1}(M_{S})\|\vec{r}_{S}\|^{2}\;+\;2\,\|\vec{r}_{S}\|\,\|\vec{r}_{S^{c}}\|+\|\vec{r}_{S^{c}}\|^{2}}{\mathrm{vol}_{2}(\mathcal{I})}}
\]
with 
\[
\lambda_{1}(M_{S})\leq\min\left(1,U_{\tau_{S}}-\sum_{i=2}^{|S|}\lambda_{i}(M_{S})\right)\leq\min(1,U_{\tau_{S}}).
\]
Smaller terms in the inequality give sharper bound.

\textbf{C2.} \textbf{Ky Fan's principle} \textbf{lower} side (sum
of \emph{smallest} eigenvalues of $B_{S}B_{S}'$\hspace{0pt}). \\
Using
$U_{\tau_{S}}=\ k-\frac{1}{\sigma_{1}^{2}}\sum_{j=n-k+1}^{n}\lambda_{j}(B_{S}B_{S}')$,
we can specialize C1 to 
\[
\cos\theta_{k}\;=\;\frac{\|\Pi_{k}\vec{r}\|}{\|\vec{r}\|}\;\leq\;\sqrt{\frac{\min(1,\ k-\frac{1}{\sigma_{1}^{2}}\sum_{j=n-k+1}^{n}\lambda_{j}(B_{S}B_{S}'))\|\vec{r}_{S}\|^{2}\;+\;2\,\|\vec{r}_{S}\|\,\|\vec{r}_{S^{c}}\|+\|\vec{r}_{S^{c}}\|^{2}}{\mathrm{vol}_{2}(\mathcal{I})}}
\]
(we write here only the largest inequality).

\textbf{C3.} Using \textbf{Cauchy interlacing} we can replace some
of the $\lambda_{j}(B_{S}B_{S}')$ in C2 with $\sigma_{j+|S|}^{2}$.
\begin{rem}
Trivial but useful: since $\Pi_{k}\succeq\Pi_{k-1}$\hspace{0pt}
(in Loewner's order), we always have $\cos\theta_{k}\ge\cos\theta_{k-1}$\hspace{0pt},
in particular $\cos\theta_{k}\ge\cos\theta_{1}$. So our lower $\Pi_{1}$-bound
already applies to all $k$. $\Pi_{k}$-bound \textbf{improves} it
whenever the top-$k$ items hold a substantial share of popularity.
\end{rem}

\subsubsection{Linear programming estimate.}

Let us prove another estimate using linear programming methods. From
the SVD decomposition of $Y=\sum_{i}\sigma_{i}\,\vec{p}_{i}\vec{q}_{i}'$\hspace{0pt}
and expansion 
\[
\vec{e}=\sum_{i}c_{i}\,\vec{p}_{i},\quad c_{i}:=\vec{p}_{i}'\vec{e}=(\vec{p},\,\vec{e}).
\]
we obtain 
\[
\vec{r}=Y'\vec{e}=\sum_{i}\sigma_{i}c_{i}\,\vec{q}_{i},\qquad\kappa_{k}=\frac{\sum_{i=1}^{k}\sigma_{i}^{2}c_{i}^{2}}{\sum_{j}\sigma_{j}^{2}c_{j}^{2}}.
\]
Set $s_{i}:=\sigma_{i}^{2}$\hspace{0pt} (descending as $s_{1}\ge s_{2}\ge\cdots\geq0$)
and define 
\[
{\displaystyle \mu\;:=\;\frac{\|\vec{r}\|^{2}}{\|\vec{e}\|^{2}}\;=\;\frac{\|Y'\vec{e}\|^{2}}{n}\;=\;\frac{\mathrm{vol}_{2}(\mathcal{I})}{vol_{0}(\mathcal{U})}=\tilde{d}_{\mathcal{I}}\frac{\mathrm{vol}_{1}(\mathcal{I})}{vol_{0}(\mathcal{U})}.}
\]
Normalize ${\displaystyle {\alpha_{i}:=\frac{c_{i}^{2}}{\sum_{j}c_{j}^{2}}=\frac{c_{i}^{2}}{n}}}$,
so $\sum_{i}\alpha_{i}=1$ and 
\[
\mu=\sum_{i}s_{i}\alpha_{i},\qquad\kappa_{k}=\frac{\sum_{i=1}^{k}s_{i}\alpha_{i}}{\mu}.
\]
Thus, in fact, $\kappa_{k}$\hspace{0pt} depends only on the spectrum
$\{s_{i}\}$ and the scalar $\mu=\|\vec{r}\|^{2}/n$.
\begin{thm*}
\textbf{1LP.} In the above notations, we have the lower and upper
bound estimates: 
\[
\cos^{2}\theta_{k}=\kappa_{k}\geq\left[\begin{array}{cc}
0, & \text{when }\sigma_{n}^{2}\leq\mu\leq\sigma_{k+1}^{2}\\
\frac{\sigma_{1}^{2}}{\mu}\frac{\mu-\sigma_{k+1}^{2}}{\sigma_{1}^{2}-\sigma_{k+1}^{2}}, & \text{when }\sigma_{k+1}^{2}<\mu\leq\sigma_{1}^{2}
\end{array}\right.
\]
\[
\cos^{2}\theta_{k}=\kappa_{k}\leq\left[\begin{array}{cc}
\frac{\sigma_{k}^{2}}{\mu}\frac{\mu-\sigma_{n}^{2}}{\sigma_{k}^{2}-\sigma_{n}^{2}}, & \text{when }\sigma_{n}^{2}\leq\mu\leq\sigma_{k}^{2}\\
1, & \text{when }\sigma_{k}^{2}<\mu\leq\sigma_{1}^{2}
\end{array}\right.
\]
\end{thm*}
\textbf{Interpretation. }If the ``degree-squared mean'' $\mu$ does
not exceed the $(k+1)$-st spectral level $\sigma_{k+1}^{2}$\hspace{0pt},
there is no information forcing $\vec{r}$ to live in the top-$k$
space (bound 0).\\
 Once $\mu$ crosses $\sigma_{k+1}^{2}$, a nonzero fraction of $\vec{r}$
must lie in the top-$k$ subspace. The bound is stronger when the
\textbf{gap} $\mu-\sigma_{k+1}^{2}$\hspace{0pt} is larger. If $\mu$
approaching $\sigma_{1}^{2}$, the bound approaches $1$: most of
$\vec{r}$ must be captured by the top-$k$ space.
\begin{itemize}
\item These are tight bounds, given only $\{\sigma_{i}\}$ and $\mu$: the
worst case is a two-point mix of the $q_{1}$\hspace{0pt} direction
and the $q_{k+1}$\hspace{0pt} direction. 
\item It ignores head structure; sometimes stronger than B-style bounds
(e.g., big spectral gap, $\mu$ near $\sigma_{1}^{2}$\hspace{0pt});
sometimes weaker (heavy head but small gap). 
\item The bounds is monotonic non-decreasing in $k$ (easy to check with
basic Calculus). Because $\sigma_{k+1}$\hspace{0pt} decreases with
$k$; as we allow a larger top subspace, $\theta_{k}$\hspace{0pt}
must shrink (alignment grows). 
\item For $k=1$ LP lower bound reduces to the ``gap-only'' principal-direction
bound:\\
 $\cos^{2}\theta_{1}\ge s_{1}\frac{(\mu-s_{2})}{[\mu(s_{1}-s_{2})]}$
if $\mu>s_{2}$\hspace{0pt}, else 0.\\
 
\item If $\sigma_{n}=0$ (common), \textbf{LP upper} bound reduces to $\cos\theta_{k}\le1$
for $\mu\le\sigma_{k}^{2}$, which is not informative. 
\item We can state the \textbf{symmetric user-side} version by replacing
$r=Y'e$, $n$ with $s=Ye$, $m$; the spectrum is the same. 
\end{itemize}

\section{Appendix.}

\subsection{A. Kumar's bounds.}

A two-sided estimate is from {[}K{]} by R. Kumar: 
\[
m+\frac{s}{\sqrt{p-1}}\leq\lambda_{1}^{2}\leq m+s\sqrt{p-1},
\]
where $m=\frac{e}{p}$, $s=\sqrt{\frac{M}{2p}-m^{2}}$, $p=\left\lfloor \frac{n+m}{2}\right\rfloor $,
$M=tr(A^{4})$, and, as we know, $\lambda_{1}=\sigma_{1}$. It is
also known that 
\[
tr(A^{4})=2|E_{G}|+4W+8c_{4}(G),
\]
where $W=\sum_{v\in V_{G}}\left(\begin{array}{c}
d_{v}\\
2
\end{array}\right)$ is the number of ``wedges'' in $G$ and $c_{4}(G)$ is the number
of 4-cycles (``butterflies'') in $G$. Lets analyze the terms in
the inequality $m+\frac{s}{\sqrt{p-1}}\leq\lambda_{1}^{2}$ . Assume
that $n+m$ is even, then $m=\frac{e}{p}=\frac{2e}{2\left\lfloor \frac{n+m}{2}\right\rfloor }=\frac{2e}{n+m}=\bar{d}$
is the average degree of a vertex. In $s^{2}=\frac{M}{2p}-m^{2}=\frac{M}{2p}-\bar{d}^{2}$,
we have 
\[
\frac{M}{2p}=\frac{2|E_{G}|}{2p}+\frac{4W}{2p}+\frac{8c_{4}}{2p}=\bar{d}+\frac{4W}{4p}+\frac{8c_{4}}{2p}.
\]
It is well known that 
\[
\sum_{i=1}^{n+m}d_{i}^{2}=\text{"number of open 2-paths"}=2e+2W
\]
and thus 
\[
\frac{2e}{2p}+\frac{4W}{2p}-m^{2}=\frac{2e+2W}{2p}+\frac{2W}{2p}-m^{2}=\left(\frac{1}{2p}\sum_{i=1}^{n+m}d_{i}^{2}-\bar{d}^{2}\right)+\frac{W}{p}=Var(d)+\frac{W}{p}
\]
and so 
\[
s^{2}=Var(d)+\frac{W}{p}+\frac{4c_{4}}{p}
\]
Thus we can say that 
\[
m+\frac{s}{\sqrt{p-1}}\cong\bar{d}+\sqrt{\frac{Var(d)+\frac{W}{p}+\frac{4c_{4}}{p}}{p-1}}.
\]

\subsection{Theorem 1C}

\begin{proof}[Proof of Lemma 1]
\mbox{}

We have for every singular vector $\vec{q}_{j}$\hspace{0pt} 
\[
\vec{q_{j}}=(\vec{q_{j}})_{S}+(\vec{q_{j}})_{S^{c}}\qquad\text{and}\qquad1=\|\vec{q_{j}}\|^{2}=\|(\vec{q_{j}})_{S}\|^{2}+\|(\vec{q_{j}})_{S^{c}}\|^{2}
\]
Since $Y'\vec{p}_{j}=\sigma_{j}\vec{q}_{j}\Rightarrow(B_{S}'\vec{p}_{j})=\sigma_{j}(\vec{q}_{j})_{S^{c}}$\hspace{0pt}.
Hence 
\[
\|(\vec{q}_{j})_{S^{c}}\|_{2}^{2}=\sigma_{j}^{-2}\|B_{S}'\vec{p}_{j}\|_{2}^{2}\le\sigma_{j}^{-2}\lambda_{\max}(B_{S}'B_{S}).
\]
Altogether we get 
\[
\|(\vec{q_{j}})_{S}\|_{2}^{2}\;\ge\;1-\frac{\lambda_{\max}(B_{S}'B_{S})}{\sigma_{j}^{2}}\,.
\]
For binary $Y$ one has the row-sum bound 
\[
\lambda_{\max}(B_{S}'B_{S})\ \le\ \sum_{i\in S^{c}}r_{i}\;=\;\operatorname{vol}(\mathcal{I})-r_{S},
\]
so

\[
\|(\vec{q}_{j})_{S}\|_{2}^{2}\;\ge\;1-\frac{\operatorname{vol}(\mathcal{I})-r_{S}}{\sigma_{j}^{2}}\,.
\]
The $(s,s)$ entry (the $s$-th diagonal element) of $\Pi_{k}$ is
\[
\pi_{s}\;:=\;[\Pi_{k}]_{ss}\;=\;\vec{e}_{s}'\Pi_{k}\vec{e}_{s}\;=(\vec{e}_{s}'Q_{k})(Q_{k}'\vec{e}_{s})=\;\sum_{j=1}^{k}(\vec{e}_{s}'\vec{q}_{j})^{2}\;=\;\sum_{j=1}^{k}q_{js}^{2},
\]
and 
\[
\tau_{S}:=\mathrm{tr}(\Pi_{k})_{SS}=\sum_{s\in S}\pi_{s}.
\]
Then 
\[
\tau_{S}=\sum_{j,s\in S}q_{js}^{2}=\sum_{j=1}^{k}\|(\vec{q}_{j})_{S}\|_{2}^{2}\geq k-(vol(\mathcal{I})-r_{S})\left(\sum_{j=1}^{k}\frac{1}{\sigma_{j}^{2}}\right)=k-\Delta_{S}H_{k},
\]
with $\Delta_{S}\LyXZeroWidthSpace:=vol(I)-r_{S}$\hspace{0pt}, and
$H_{k}:=\sum_{j=1}^{k}\sigma_{j}^{-2}$. 
\end{proof}

\begin{proof}[Proof of Lemma $1'$]
\mbox{}

Since $B_{S}'p_{j}=\sigma_{j}\,(q_{j})_{S^{c}}$, we can write the
norm as 
\[
\|(q_{j})_{S^{c}}\|^{2}=\frac{1}{\sigma_{j}^{2}}\|B_{S}'p_{j}\|^{2}=\frac{1}{\sigma_{j}^{2}}\,(B_{S}'p_{j})'(B_{S}'p_{j})=\frac{1}{\sigma_{j}^{2}}\,p_{j}'(B_{S}B_{S}')p_{j}.
\]
Let $P_{k}^{(u)}=\sum_{j=1}^{k}p_{j}p_{j}'$\hspace{0pt} be the projector
onto $\mathrm{span}\{p_{1},\dots,p_{k}\}$, and $W_{k}=\sum_{j=1}^{k}\frac{1}{\sigma_{j}^{2}}\,p_{j}p_{j}'$
is projection onto the same span with scaling. We have 
\[
P_{k}^{(u)}W_{k}=W_{k}P_{k}^{(u)}=W_{k}.
\]
Now summing $\|(q_{j})_{S^{c}}\|^{2}$ over $j=1,\dots,k$ gives 
\begin{align}
\sum_{j=1}^{k}\|(q_{j})_{S^{c}}\|^{2}=\sum_{j=1}^{k}\frac{1}{\sigma_{j}^{2}}p_{j}'(B_{S}B_{S}')p_{j}=\mathrm{Tr}\left(\sum_{j=1}^{k}\frac{1}{\sigma_{j}^{2}}p_{j}'(B_{S}B_{S}')p_{j}\right)=\\
=\sum_{j=1}^{k}\frac{1}{\sigma_{j}^{2}}\mathrm{Tr}\left(p_{j}'(B_{S}B_{S}')p_{j}\right)=\sum_{j=1}^{k}\frac{1}{\sigma_{j}^{2}}\mathrm{Tr}\left((B_{S}B_{S}')p_{j}p_{j}'\right)=\\
=\mathrm{Tr}\left(B_{S}B_{S}'\left(\sum_{j=1}^{k}\frac{1}{\sigma_{j}^{2}}p_{j}p_{j}'\right)\right)=\mathrm{Tr}\!\Big(B_{S}B_{S}'\,W_{k}\Big)=\\
=\mathrm{Tr}\!\Big(B_{S}B_{S}'\,P_{k}^{(u)}\,W_{k}P_{k}^{(u)}\,\Big)=\mathrm{Tr}\!\Big(P_{k}^{(u)}\,B_{S}B_{S}'\,P_{k}^{(u)}\,W_{k}\Big).
\end{align}
Note $P_{k}^{(u)}B_{S}B_{S}'P_{k}^{(u)}\succeq0$ (in Loewner's order)
because $B_{S}B_{S}'\succeq0$. On $\mathrm{span}\{p_{1},\dots,p_{k}\}$,
$W_{k}$\hspace{0pt} is diagonal with eigenvalues $1/\sigma_{j}^{2}$\hspace{0pt}
(nonincreasing in $j$). Hence 
\[
\frac{1}{\sigma_{1}^{2}}P_{k}^{(u)}\;\preceq\;W_{k}\;\preceq\;\frac{1}{\sigma_{k}^{2}}P_{k}^{(u)}.
\]
Therefore 
\[
\frac{1}{\sigma_{1}^{2}}\mathrm{Tr}\!\Big(P_{k}^{(u)}\,B_{S}B_{S}'\,P_{k}^{(u)}\,P_{k}^{(u)}\Big)\leq\mathrm{Tr}\!\Big(P_{k}^{(u)}\,B_{S}B_{S}'\,P_{k}^{(u)}\,W_{k}\Big)\leq\frac{1}{\sigma_{k}^{2}}\mathrm{Tr}\!\Big(P_{k}^{(u)}\,B_{S}B_{S}'\,P_{k}^{(u)}\,P_{k}^{(u)}\Big).
\]
Using idempotence of projectors and trace properties, this inequality
changes to 
\[
\frac{1}{\sigma_{1}^{2}}\,\mathrm{Tr}\!\big(P_{k}^{(u)}\,B_{S}B_{S}'\big)\leq\sum_{j=1}^{k}\|(q_{j})_{S^{c}}\|^{2}\leq\frac{1}{\sigma_{k}^{2}}\,\mathrm{Tr}\!\big(P_{k}^{(u)}\,B_{S}B_{S}'\big)
\]
and using Ky Fan's principles, 
\[
\frac{1}{\sigma_{1}^{2}}\sum_{j=n-k+1}^{n}\lambda_{j}(B_{S}B_{S}')\;\leq\;\sum_{j=1}^{k}\|(q_{j})_{S^{c}}\|^{2}\leq\;\frac{1}{\sigma_{k}^{2}}\sum_{j=1}^{k}\lambda_{j}(B_{S}B_{S}').
\]
Equivalently, 
\[
k-\frac{1}{\sigma_{1}^{2}}\sum_{j=n-k+1}^{n}\lambda_{j}(B_{S}B_{S}')\;\geq\;\sum_{j=1}^{k}\|(q_{j})_{S}\|^{2}\;\ge\;k-\frac{1}{\sigma_{k}^{2}}\sum_{j=1}^{k}\lambda_{j}(B_{S}'B_{S})\ \tag{KF–S}.
\]
\end{proof}

\begin{proof}[Proof of Theorem 1C]
\mbox{}

\textbf{Step 1. Spectral bounds on the principal submatrix.} 
Consider
the principal submatrix $M_{S}:=(\Pi_{k})_{SS}\in\mathbb{R}^{|S|\times|S|}$
(cuts out $S$-s rows and columns). Let $E_{S}:\mathbb{R}^{|S|}\to\mathbb{R}^{m}$
be the column-selector (embeds vectors supported on $S$). Then 
\[
M_{S}=E_{S}'\Pi_{k}E_{S}=E_{S}'Q_{k}Q_{k}'E_{S}=(Q_{k}'E_{S})'(Q_{k}'E_{S}),
\]
so it is positive semi-definite. Also $\|\Pi_{k}\|_{2}=1$ (it's a
projector), hence 
\[
(M_{S}x,x)=(\Pi_{k}E_{S}x,E_{S}x)\leq\|E_{S}x\|_{2}^{2}=\|x\|_{2}^{2},
\]
so $\|M_{S}\|_{2}\leq1$. Therefore all eigenvalues of $M_{S}$\hspace{0pt}
lie in $[0,1]$. Note that - $\mathrm{tr}(M_{S})=\tau_{S}=\sum_{s\in S}[\Pi_{k}]_{ss}=\sum_{j=1}^{k}\|(q_{j})_{S}\|^{2}$
- $\mathrm{rank}(M_{S})\le\mathrm{rank}(\Pi_{k})=k$. - In particular,
if $|S|>k$, then $M_{S}$\hspace{0pt} has at least $|S|-k$ zero
eigenvalues. Let the eigenvalues of $M_{S}$\hspace{0pt} be $\lambda_{1}(M_{S})\ge\cdots\ge\lambda_{|S|}(M_{S})\ge0$,
each $\le1$, with sum $\sum_{i=1}^{|S|}\lambda_{i}(M_{S})=\tau_{S}$\hspace{0pt}.
Then 
\begin{align}
\lambda_{min}(M_{S})=\lambda_{|S|}(M_{S})= & \tau_{S}-\sum_{i=1}^{|S|-1}\lambda_{i}(M_{S})\geq\max\left(0,L_{\tau_{S}}-\sum_{i=1}^{|S|-1}\lambda_{i}(M_{S})\right)\\
\geq & \max\left(0,L_{\tau_{S}}-(|S|-1)\right)\equiv[L_{\tau_{S}}-(|S|-1)]_{+}.
\end{align}
\begin{align}
\lambda_{max}(M_{S})=\lambda_{1}(M_{S})=\tau_{S}-\sum_{i=2}^{|S|}\lambda_{i}(M_{S})\\
\leq\min\left(1,U_{\tau_{S}}-\sum_{i=2}^{|S|}\lambda_{i}(M_{S})\right)\leq\min(1,U_{\tau_{S}})
\end{align}
When $|S|>k$, due to the previous remark $\lambda_{\min}(M_{S})=0$,
consistent with $0=[\,\tau_{S}-(|S|-1)\,]_{+}$ because $\tau_{S}\le k\le|S|-1$.

\textbf{Step 2. Block lower bound for $r'\Pi_{k}r$.}

Let $(\vec{r}_{S},\vec{r}_{S^{c}})$ denote the rearrangement of coordinates
of vector $\vec{r}$ such that the coordinates from $S$ come first,
and let $P$ be the permutation matrix such that $P\vec{r}=(\vec{r}_{S},\vec{r}_{S^{c}})$.
With $M_{S}=(\Pi_{k})_{SS}$, $C=(\Pi_{k})_{SS^{c}}$, $N=(\Pi_{k})_{S^{c}S^{c}}$,
the projection $\Pi_{k}$ in this permuted basis has the form 
\[
\Pi_{k'}=P\Pi_{k}P'=\begin{bmatrix}M_{S} & C\\
C' & N
\end{bmatrix},\qquad Pr=\begin{bmatrix}\vec{r}_{S}\\
\vec{r}_{S^{c}}
\end{bmatrix}
\]
and 
\[
\vec{r}'\Pi_{k}\vec{r}=(P\vec{r})'\Pi_{k'}(P\vec{r})=\vec{r}_{S}'M_{S}\vec{r}_{S}+2\,\vec{r}_{S}'C\vec{r}_{S^{c}}+\vec{r}_{S^{c}}'N\vec{r}_{S^{c}},
\]
since everything is real and 
\[
\vec{r}_{S^{c}}'C'\vec{r}_{S}=(C'\vec{r}_{S},\vec{r}_{S^{c}})=(\vec{r}_{S},C\vec{r}_{S^{c}})=\vec{r}_{S}'C\vec{r}_{S^{c}}
\]
 is the same scalar. Using $M_{S}\succeq\lambda_{\min}(M_{S})I$,
$N\succeq0$, and $\|C\|_{2},\|N\|_{2}\leq1$ (since $\|\Pi_{k}\|_{2}=1$),
we get that 
\[
\begin{array}{ccccc}
\lambda_{\min}(M_{S})\|\vec{r}_{S}\|^{2} & \leq & \vec{r}_{S}'M_{S}\vec{r}_{S} & \leq & \lambda_{\max}(M_{S})\|\vec{r}_{S}\|^{2}\\
-\,2\,\|\vec{r}_{S}\|\,\|\vec{r}_{S^{c}}\| & \leq & 2\,\vec{r}_{S}'C\vec{r}_{S^{c}}\; & \leq\; & \,2\,\|\vec{r}_{S}\|\,\|\vec{r}_{S^{c}}\|\\
0 & \leq & \vec{r}_{S^{c}}'N\vec{r}_{S^{c}} & \leq & \|\vec{r}_{S^{c}}\|^{2}
\end{array}
\]

and 
\[
\begin{array}{cc}
\vec{r}'\Pi_{k}\vec{r}\;\geq\; & \lambda_{\min}(M_{S})\,\|\vec{r}_{S}\|^{2}\;-\;2\,\|\vec{r}_{S}\|\,\|\vec{r}_{S^{c}}\|\\
\vec{r}'\Pi_{k}\vec{r}\;\leq\; & \lambda_{\max}(M_{S})\,\|\vec{r}_{S}\|^{2}\;+\;2\,\|\vec{r}_{S}\|\,\|\vec{r}_{S^{c}}\|\;+\;\|\vec{r}_{S^{c}}\|^{2}
\end{array}
\]

Combining with (†) and the bound on $\tau_{S}$: 
\[
\quad\vec{r}'\Pi_{k}\vec{r}\;\geq\;\Big(\,\Big[L_{\tau_{S}}-(|S|-1)\Big]_{+}\Big)\,\|\vec{r}_{S}\|^{2}\;-\;2\,\|\vec{r}_{S}\|\,\|\vec{r}_{S^{c}}\|\quad
\]
\[
\quad\vec{r}'\Pi_{k}\vec{r}\;\leq\;\Big(\,\min[1,U_{\tau_{S}}]\Big)\,\|\vec{r}_{S}\|^{2}\;+\;2\,\|\vec{r}_{S}\|\,\|\vec{r}_{S^{c}}\|\;+\;\|\vec{r}_{S^{c}}\|^{2}\quad
\]
\end{proof}

\subsection{LP theorem.}

\hspace*{\fill}

Finding a lower and an upper bounds for $\kappa_{k}$\hspace{0pt}
turns into two linear programming problems on $\alpha_{i}$: 
\[
\begin{cases}
N(\alpha)=\sum_{i\le k}s_{i}\alpha_{i}\to\min\\
N(\alpha)=\sum_{i\le k}s_{i}\alpha_{i}\to\max
\end{cases}\quad\text{s.t.}\quad\sum_{i=1}^{n}\alpha_{i}=1,\;\sum_{i=1}^{n}s_{i}\alpha_{i}=\mu,\;\alpha_{i}\ge0,\;s_{1}\ge s_{2}\ge\cdots\geq0.\tag{LP}
\]

The feasible region is a polytope - the intersection of the simplex
$\{\alpha\ge0,\ \sum\alpha_{i}=1\}$ with the hyperplane $\sum s_{i}\alpha_{i}=\mu$.
It's nonempty and compact. Lets also note that, as a convex combination
of (ordered values) $s_{i}$, $\mu\in[s_{n},s_{1}]$. 

Linear programs attain their optima at extreme points (a.k.a. vertices)
of the feasible polytope. At a vertex, at least $m-2$ of the $\alpha_{i}$\hspace{0pt}'s
must be active ($\alpha_{i}=0$). Therefore at most 2 entries can
be positive. (If $\mu$ happens to equal some $s_{j}$\hspace{0pt},
there is just 1 positive entry: $\alpha_{j}=1$) 

Define the index sets $A=\{1,\dots,k\}$ and $B=\{k+1,\dots,n\}$.
Since there are at most two positive~$\alpha_{i}$\hspace{0pt}'s,
there are only three cases: 

Both indices in\textbf{~$A$}: Then~$N(\alpha)=\mu$.

Both indices in\textbf{~$B$}: Then~$N(\alpha)=0$, which is ideal.
This is possible if $\mu$~lies between two values in~$B$. Since~$B=\{k+1,\ldots,n\}$,
the largest value for $\mu$ to take~is~$s_{k+1}\LyXZeroWidthSpace$.
So 
\[
\text{if }s_{k+1}\LyXZeroWidthSpace\geq\mu,\text{ then we can achieve  }N(\alpha)=0.
\]

One index in~$A$, one in\textbf{~$B$}: This is the most interesting
case. Suppose~$i\in A$,~$j\in B$, and 
\[
\alpha_{i}+\alpha_{j}=1,\quad s_{i}\alpha_{i}+s_{j}\alpha_{j}=\mu.
\]

Solving this system gives: 
\[
\alpha_{i}=\frac{\mu-s_{j}}{s_{i}-s_{j}},\quad\alpha_{j}=\frac{s_{i}-\mu}{s_{i}-s_{j}}.
\]

For these to be nonnegative, we require $s_{j}\le\mu\le s_{i}$. The
objective becomes: 
\[
N(\alpha)=s_{i}\alpha_{i}=s_{i}\frac{\mu-s_{j}}{s_{i}-s_{j}}.
\]

This is a decreasing function both in~$s_{i}$\hspace{0pt}~and
in~$s_{j}$. 

To minimize~$N(\alpha)$, we want to choose: 
\begin{itemize}
\item The largest possible~$s_{i}$~from~$A$:~$s_{1}$\hspace{0pt}
\item The largest possible~$s_{j}$~from~$B$:~$s_{k+1}$.\hspace{0pt}
\end{itemize}
Then the solution is: 
\[
\alpha_{1}=\frac{\mu-s_{k+1}}{s_{1}-s_{k+1}},\quad\alpha_{k+1}=\frac{s_{1}-\mu}{s_{1}-s_{k+1}},
\]
feasible when $s_{1}\geq\mu\geq s_{k+1}$, with 
\[
N(\alpha)=s_{1}\frac{\mu-s_{k+1}}{s_{1}-s_{k+1}}.
\]
To maximize~$N(\alpha)$, we want to choose: 
\begin{itemize}
\item The smallest possible~$s_{i}$~from~$A$:~$s_{k}$\hspace{0pt}
\item The smallest possible~$s_{j}$~from~$B$:~$s_{n}$.\hspace{0pt}
\end{itemize}
Then the solution is: 
\[
\alpha_{k}=\frac{\mu-s_{n}}{s_{k}-s_{n}},\quad\alpha_{n}=\frac{s_{k}-\mu}{s_{k}-s_{n}},
\]
feasible when $s_{k}\geq\mu\geq s_{n}$, with 
\[
N(\alpha)=s_{k}\frac{\mu-s_{n}}{s_{k}-s_{n}}.
\]
All together: 
\begin{itemize}
\item when $s_{n}\leq\mu\leq s_{k+1}$ we have vertices with $N(\alpha)=0$
\item when $s_{k+1}<\mu\leq s_{1}$ we are choosing between (not minimal)
$\mu$ and $s_{1}\frac{\mu-s_{k+1}}{s_{1}-s_{k+1}}.$
\end{itemize}
So the final answer to \textbf{(LP)} is 
\[
N(\alpha)\geq\left[\begin{array}{cc}
0, & \text{when }s_{n}\leq\mu\leq s_{k+1}\\
s_{1}\frac{\mu-s_{k+1}}{s_{1}-s_{k+1}}, & \text{when }s_{k+1}<\mu\leq s_{1}
\end{array}\right.\qquad\text{and }\qquad N(\alpha)\leq\left[\begin{array}{cc}
s_{k}\frac{\mu-s_{n}}{s_{k}-s_{n}}, & \text{when }s_{n}<\mu\leq s_{k}\\
\mu, & \text{when }s_{k}\leq\mu\leq s_{1}
\end{array}\right.
\]

And so 
\[
\cos^{2}\theta_{k}=\kappa_{k}\geq\left[\begin{array}{cc}
0, & \text{if }\sigma_{n}^{2}\leq\mu\leq\sigma_{k+1}^{2}\\
\frac{\sigma_{1}^{2}}{\mu}\frac{\mu-\sigma_{k+1}^{2}}{\sigma_{1}^{2}-\sigma_{k+1}^{2}}, & \text{if }\sigma_{k+1}^{2}<\mu\leq\sigma_{1}^{2}
\end{array}\right.
\]
\[
\cos^{2}\theta_{k}=\kappa_{k}\leq\left[\begin{array}{cc}
\frac{\sigma_{k}^{2}}{\mu}\frac{\mu-\sigma_{n}^{2}}{\sigma_{k}^{2}-\sigma_{n}^{2}}, & \text{when }\sigma_{n}^{2}\leq\mu\leq\sigma_{k}^{2}\\
1, & \text{when }\sigma_{k}^{2}<\mu\leq\sigma_{1}^{2}
\end{array}\right.
\]

\end{document}